\documentclass[12pt]{article}
\usepackage{arxiv}

\usepackage{amsmath}
\usepackage{amsthm}
\usepackage{subcaption}

\author{Tyler Meadows \\Department of Mathematics,\\
           University of Idaho,
           Moscow, Idaho 83844, USA\\
         \And Elissa J. Schwartz\\Department of Mathematics and Statistics, and School of Biological Sciences,\\ Washington State University, Pullman, Washington 99164, USA\\
}
\usepackage{natbib}
\usepackage{graphicx}
\newtheorem{theorem}{Theorem}
\newtheorem{lemma}[theorem]{Lemma}
\newtheorem{prop}[theorem]{Proposition}

\title{A Model of Virus Infection with Immune Responses Supports Boosting CTL Response to Balance Antibody Response}
\renewcommand{\shorttitle}{Analysis of A Virus Infection Model with Immune Responses}

\begin{document}

\maketitle

\begin{abstract}
We analyze a within-host model of virus infection with antibody and CD$8^+$ cytotoxic T lymphocyte (CTL) responses proposed by Schwartz {\it et al.} (2013).  The goal of this work is to gain an overview of the stability of the biologically-relevant equilibria as a function of the model's immune response parameters.  We show that the equilibria undergo at most two forward transcritical bifurcations.  The model is also explored numerically and results are applied to equine infectious anemia virus infection. In order to arrive at stability of the biologically-relevant endemic equilibrium characterized by coexistence of antibody and CTL responses, the parameters promoting CTL responses need to be boosted over parameters promoting antibody production. This result may seem counter-intuitive (in that a weaker antibody response is better) but can be understood in terms of a balance between CTL and antibody responses that is needed to permit existence of CTLs. In conclusion, an intervention such as a vaccine that is intended to control a persistent viral infection with both immune responses should moderate the antibody response to allow for stimulation of the CTL response.
\end{abstract}


\section{Introduction}
\label{intro}

Equine infectious anemia virus (EIAV) is an infection in horses that is similar to HIV (human immunodeficiency virus) in structure, genome, and life cycle \citep{Lerouxetal2004}. However, horses infected with EIAV do not develop AIDS, as do HIV-infected individuals without treatment.  Instead, EIAV-infected horses produce immune responses that control the infection \citep{CraigoMontelaro2013}. This control has been shown to be mediated by both CD$8^+$ cytotoxic T lymphocytes (CTLs) and antibody responses, which ultimately limit virus replication and prevent symptoms in long-term infected animals, even though the virus infection is not cleared \citep{Cooketal2013}. Specifically, T cell epitopes that indicate a broadening of CTL response have been identified to persist in long-term infection \citep{Issel2014, McGuireetal2000, Tagmyer2007}. Furthermore, evolving broadly neutralizing antibodies have been found to be needed to maintain asymptomatic disease \citep{Hammondetal1997, Sponseller2007, Craigo2007}. Thus, EIAV infection has been the subject of many controlled experiments to investigate the immune response to infection \citep{Mealeyetal2008, Tayloretal2010, Tayloretal2011, Schwartzetal2015, Schwartzetal2018}. The more we understand how CTLs and antibodies control EIAV infection, the more insight will be gained on how best to develop effective interventions that control other similar viral infections.

The standard model of virus infection \citep{NowakBangham1996, Perelsonetal1996} is a system of three ordinary differential equations (ODEs) that can account for many experimental observations in the stages of both HIV and equine infectious anemia virus (EIAV) infection \citep{Schwartzetal2018, PerelsonRibeiro2013, Stafford2000, Phillips1996, Noecker2015}. This model depicts the concentrations of uninfected cells, infected cells, and virus particles, and represents a virus infecting a target cell population, with the infected cells then producing more viral particles. This model does not include additional equations that explicitly  model immune responses. Since immune responses are dynamic as well, further realism can be gained by explicitly including equations representing immune responses.

\citet{NowakBangham1996} developed a model consisting of a system of four ODEs that includes the dynamics of one component of the immune system: the population of CTLs that kill infected cells.  A subsequent model by \citet{Wodarz2003} presents a system of five equations to model an infection by hepatitis C virus; this model explicitly includes immune responses given by populations of both CTLs and antibodies.

In 2013, \citet{Schwartzetal2013} published a five-equation model of EIAV infection that also includes two populations of immune responses, one for CTLs and one for antibodies.  This mathematical model differs from that of \citet{Wodarz2003} in the equation describing the antibody response. \citet{Schwartzetal2013} depict antibody production as proportional to the concentration of virus, rather than proportional to the interactions between viruses and pre-existing antibodies.
Other authors have built more complexity into this equation, such as by including B cell dynamics and differentiation into antibody producing cells \citep{LeMillerGanusov2015}, but this approach necessarily relies upon the addition of more parameters with unknown values in the case of EIAV infection.
The \citet{Schwartzetal2013} model takes a step back and uses a
more straightforward
approach, in which antibody production is modeled as first order in $V$. This choice still captures the essence of the biology, given that antibody production is correlated with the quantity of virus \citep{CraigoMontelaro2013,Sajadi2011,Koopman2015}.

In the notation of \citet{Schwartzetal2013}, the system of equations is

\begin{subequations}\label{eqn:EIAVmodel}
\begin{align}
\dot{M} &= \lambda - \rho M - \beta M V
\label{Mdot},\\
\dot{I}&=\beta M V - \delta I - k I C
\label{Idot},\\
\dot{V} &= b I - \gamma V - f V A
\label{Vdot},\\
\dot{C} &= \psi I C - \omega C
\label{Cdot},\\
\dot{A} &= \alpha V - \mu A
\label{Adot}.
\end{align}
\end{subequations}
The state variables are the concentrations of uninfected cells $M$ (in this infection, uninfected cells are macrophages, a type of white blood cell that is the target cell of EIAV), infected cells $I$, virus $V$, cytotoxic T lymphocytes $C$, and antibodies $A$.  Solutions of the EIAV model are only biologically relevant when all of the state variables are non-negative.  (We use ``biological'' as a synonym for ``non-negative''.)  The lower-case Greek and Latin letters denote 12 parameters, which are assumed to be positive.  These equations are interpreted as follows:  Eq. \eqref{Mdot} describes uninfected cells introduced at rate $\lambda$, removed at rate $\rho M$ and infected by virus particles at rate $\beta M V$;  Eq. \eqref{Idot} describes infected cells produced at rate $\beta M V$, removed at rate $\delta I$, and killed by CTLs at rate $k I C$;  Eq. \eqref{Vdot} describes virus particles (measured in viral RNA, vRNA) produced by infected cells at rate $b I$, removed at rate $\gamma V$, and neutralized by interaction with antibodies at rate $f V A$;  Eq. \eqref{Cdot} describes CTLs produced at rate $\psi I C$ and removed at rate $\omega C$, as in \citet{NowakBangham1996, Wodarz2003}; and finally, Eq. \eqref{Adot} describes antibody molecules produced at rate $\alpha V$ and removed at rate $\mu A$.

The model has five equilibria, yet only three of these can have non-negative values for all of the state variables and thus be biologically relevant.  These equilibrium states are (1) the infection-free equilibrium (IFE) $\mathbf{E}_0$, (2) the antibody-only equilibrium $\mathbf{E}_1$ describing an infection limited by an antibody response but not a CTL response, and (3) the coexistence equilibrium $\mathbf{E}_3$ describing an infection limited by both an antibody response and a CTL response.

\citet{Schwartzetal2013} showed that the existence of the biologically relevant equilibria could be determined by the basic reproduction number $R_0$ and a second threshold $R_1$. They identified example parameter sets corresponding to scenarios where each equilibrium is stable. However, they did not show that the thresholds determine the stability of the equilibria. In this paper we build upon the results of \citet{Schwartzetal2013} by investigating the bifurcation structure of system \eqref{eqn:EIAVmodel}, and by showing that there are at most two forward transcritical bifurcations. We use a combination of standard techniques to analyze the dynamics of \eqref{eqn:EIAVmodel} in general and verify that the thresholds are associated with bifurcation points.

We show that the infection free equilibrium $\mathbf{E}_0$ is globally asymptotically stable when $R_0<1$ and unstable when $R_0>1$. The antibody-only equilibrium $\mathbf{E}_1$ is locally asymptotically stable when $R_1<1<R_0$, indicating that there are bifurcations when $R_0=1$ and $R_1=1$. In particular, we use the next generation matrix method \citep{Diekmannetal1990,vandenDriesscheWatmough2002,vandenDriessche2017} to show that the $R_1=1$ bifurcation is a forward transcritical bifurcation between $\mathbf{E}_1$ and $\mathbf{E}_3$. Our analysis is supported by  bifurcation diagrams that can also be used to help identify parameter ranges associated with the stability of these equilibria.

The biological goal of our analysis is to identify the key parameters that correspond to the stability of each equilibrium. In particular, since EIAV infection is controlled by both antibodies and CTLs, we aim to determine which parameter ranges lead the system to the coexistence equilibrium $\mathbf{E}_3$. We can then use this information to determine the specific parameters to target with an intervention such as a vaccine.  For example, using the CTL production rate ($\psi$) or antibody production rate ($\alpha$) as control parameters, we can identify which ranges would be needed by a vaccine that stimulates CTL production or antibody production sufficiently to drive the system to $\mathbf{E}_3$. Such advances in our understanding of the modes of action of the immune system that control EIAV could indicate the targets for a potential vaccine to control other infections, like HIV.

\section{Analytical Results}
\label{sec:1}
\begin{prop}
If $(M,I,V,C,A)$ is a solution to the EIAV model, \eqref{eqn:EIAVmodel} with biological initial conditions, then $(M,I,V,C,A)$ remains biological for all time. Furthermore, $(M,I,V,C,A)$ is bounded for all time. \label{prop:biological}
\end{prop}

\begin{proof}
First, suppose that $M(t)$ becomes negative, then there exists $t_0\geq 0$ such that $M(t_0)=0$. At $t=t_0$, $\dot{M} = \lambda>0$, and so there exists $\epsilon>0$ such that $M(t)>0$ for all $t\in(t_0,t_0+\epsilon)$. Thus $M(t)$ cannot become negative. Suppose now that $C(t)$ becomes negative, then there exists $t_0$ such that $C(t_0)=0$. At $t=t_0$ $\dot{C}=0$, and therefore $C(t)=0$ for all $t$ by the Picard-Lindel\"of theorem. Thus, if $C(0)>0$, then $C(t)>0$ for all $t$.

Suppose that one of $I$ or $V$ becomes negative. Then there exists $t^*\geq0$ such that $\min\{I(t^*),V(t^*)\} =0$. If $I(t^*) =0$ and $V(t^*)>0$, then $\dot{I}(t^*) = \beta M(t^*)V(t^*)>0$, and thus there exists $\epsilon>0$ such that $I(t)>0$ for all $t\in(t^*,t^*+\epsilon)$. If $I(t^*)>0$ and $V(t^*)=0$, then $\dot{V}(t^*) = bI(t^*)>0$, and therefore there exists $\epsilon>0$ such that $V(t)>0$ for all $t\in(t^*,t^*+\epsilon)$. Finally, if $I(t^*)=V(t^*)=0$, then $\dot{I}(t^*)=\dot{V}(t^*)=0$, and thus $I(t)=V(t)=0$ for all $t>0$. It follows that $I(t)\geq0$ and $V(t)\geq0$ for all $t\geq0$.

Finally, suppose that $A(t)$ becomes negative, then there exists $t_0\geq0$ such that $A(t_0)=0$. At $t=t_0$ $\dot{A}(t_0)=\alpha V(t_0)\geq0$.
It follows that $A(t)>0$ for all $t\geq0$.

To show that solutions are bounded, consider $\Sigma(t) := M(t) +I(t) +\tfrac{k}{\psi}C(t)$, which satisfies

\begin{align}
    \dot\Sigma(t) &= \lambda - \rho M(t) -\delta I(t)  -\frac{\omega k}{\psi}C(t),\notag\\
    &\leq \lambda - \kappa \Sigma(t),\label{eqn:prop1intermediate}
\end{align}
  where $\kappa= \min\{\rho,\delta,\omega\}$. By multiplying eq. \eqref{eqn:prop1intermediate} by $e^{\kappa t}$, rearranging, and integrating, we find

\begin{align*}
\Sigma(t) &\leq \frac{\lambda}{\kappa}+\Sigma(0)e^{-\kappa t}-\frac{\lambda}{\kappa}e^{-\kappa t},\\
&\leq \max\left\{\frac{\lambda}{\kappa},\Sigma(0)\right\}.
\end{align*}
Since $M(t)$, $I(t)$, and $C(t)$ are positive, it follows that they are bounded in forward time.
Since $I(t)$ is bounded above, there exists $T>0$ such that

\begin{align*}
\dot{V}(t) &\leq T-\gamma V(t) - f V(t) A(t), \\
        &\leq T-\gamma V(t),
\end{align*}
where the second inequality follows from the fact that $fV(t)A(t)\geq0$. It similarly follows that $V(t)$ is bounded for all $t\geq0$. Finally, since $I(t)$ is bounded above, there exists $Q>0$ such that

\begin{align*}
    \dot{A}(t) \leq Q-\mu A(t),
\end{align*}
and a similar argument shows that $A(t)$ is bounded for all $t\geq0$.
\end{proof}

\subsection{Equilibria of the EIAV Model}
\label{sec:2}

The equilibria of the EIAV model, \eqref{eqn:EIAVmodel}, are derived in \citet{Schwartzetal2013}. We report the non-negative ones here for convenience.

The infection free equilibrium (IFE) is given by
\begin{equation}
\mathbf{E}_0 = (M_0, I_0, V_0, C_0, A_0) = (\lambda/\rho,0,0,0,0).
\label{E0}
\end{equation}
Only uninfected cells are present.

The boundary equilibrium is given by $\mathbf{E}_1 = (M_1, I_1, V_1, C_1, A_1)$ where

\begin{subequations}
\begin{align}
M_1 &= \frac{\lambda}{(\rho + \beta V_1)},
\label{M1}\\
I_1 &= \frac{\lambda}{\delta} \cdot \frac{\beta V_1}{\rho + \beta V_1},
\label{I1}\\
V_1 &= \frac{-(\alpha f \rho + \beta \gamma \mu)+\sqrt{(\alpha f \rho - \beta \gamma \mu)^2 + \frac{4 \alpha b \beta^2 f \mu\lambda}{\delta}}}{2 \alpha \beta f},
\label{V1}\\
C_1 &= 0,\label{C1}\\
A_1 &= \frac{\alpha}{\mu} V_1.\label{A1}
\end{align}
\end{subequations}
The boundary equilibrium describes an infection---with both virus particles and infected cells present---that elicits an antibody response but not a CTL response.

Finally, the endemic equilibrium is given by $\mathbf{E}_3=(M_3,I_3,V_3,C_3,A_3)$ where

\begin{subequations}
\begin{align}
M_3 &= \frac{\lambda}{\rho + \beta V_3},
\label{M3}\\
I_3 &= \omega/\psi,
\label{I3}\\
V_3&=\frac{-\gamma \mu \psi + \sqrt{(\gamma \mu \psi)^2 + 4 \alpha b f \mu \psi \omega}}{2 \alpha f \psi},
\label{V3}\\
C_3 &= \frac{\lambda \psi}{k \omega} \cdot \frac{\beta V_3}{\rho + \beta V_3} - \frac{\delta}{k},
\label{C3}\\
A_3 &= \frac{\alpha}{\mu}V_3.
\label{A3}
\end{align}
\end{subequations}
The endemic equilibrium describes an infection that elicits both antibody and CTL responses from the immune system.  Note that this equilibrium is only non-negative when $C_3 \geq 0$.

\subsection{Stability Analysis}
\label{sec:3}

\citet{Schwartzetal2013} made use of the next generation matrix method \citep{vandenDriesscheWatmough2002} to analyze the linear stability of the IFE.  We combine their results in the statement of Theorem \ref{Schwartz} below. If an infected cell is introduced into an IFE, the basic reproduction number $R_0$ is roughly the average number of infected cells produced.  The number is a threshold value. {\it I.e.},  if $R_0<1$, then the infection will die out and if $R_0>1$, then the infection will grow.

 \begin{theorem} Consider the EIAV model with positive parameter values.  The basic reproduction number is
 \begin{equation}
 R_0 = \frac{b \beta \lambda}{\delta \gamma \rho}.
 \label{R_0}
 \end{equation}
  Further, the disease-free equilibrium $\mathbf{E}_0$ is globally asymptotically stable for $R_0<1$ and unstable for $R_0>1$.
 \label{Schwartz}
 \end{theorem}
 \begin{proof}
 Linear stability results for $\mathbf{E}_0$ and a calculation of $R_0$ using the next generation matrix method were done in \citet{Schwartzetal2013}.

 Suppose that $R_0<1$, then there exists $\epsilon>0$ such that $R_0+\epsilon \leq 1.$ To show global stability of $\mathbf{E}_0$, define the Lyapunov function

 \begin{equation*}
    U_{\epsilon} = M(t)-M_0\log\left(\frac{M(t)}{M_0}\right)+I(t)+\frac{\delta}{\beta}V(t)+\frac{k}{\psi}C(t)+\epsilon\frac{\gamma\delta}{b\alpha}A(t),
 \end{equation*}
 so that, for any $\epsilon>0$, $U_\epsilon\geq0$ with $U_\epsilon = 0$ only when $(M,I,V,C,A) = \mathbf{E}_0$.
 The derivative of $U_\epsilon$ along trajectories is

 \begin{align*}
     \dot{U}_{\epsilon} &= -\frac{\rho}{M(t)}(M_0-M(t))^2 -\frac{\gamma\delta}{b}(1-R_0-\epsilon)V(t)\\&\quad-\gamma\delta fV(t)A(t)-\frac{k\omega}{\psi}C(t)-\epsilon\frac{\gamma\delta\mu}{b\alpha}A(t)
     \leq 0.
 \end{align*}
 By LaSalle's invariance principle, solutions converge to the largest compact invariant set of \eqref{eqn:EIAVmodel} that is contained in $ \{(M(t),I(t),V(t),C(t),A(t)): \dot{U}_\epsilon = 0\}= \{(M_0,I(t),0,0,0)\}$. {\it I.e.,} $\mathbf{E}_0.$
 \end{proof}

Note that $V_1$ may be written in terms of $R_0$ as
\begin{equation}\label{eq:V1}
V_1 = \frac{-(\alpha f \rho + \beta \gamma \mu)+\sqrt{(\alpha f \rho + \beta \gamma \mu)^2 + 4 \alpha f \rho \beta \gamma \mu(R_0-1)}}{2 \alpha \beta f},
\end{equation}
so that when $R_0=1$, $V_1 = 0$. It is easy to see from this that $\mathbf{E}_0=\mathbf{E}_1$ when $R_0=1$.

The antibody-only equilibrium $\mathbf{E}_1$ represents a viral infection with an antibody response, $A_1>0$, but no CTL response, $C_1=0$.  We may treat the CTL response as an active variable, similar to the use of ``infectious variable'' in the terminology of \citet{vandenDriesscheWatmough2002}. From this new point of view the boundary equilibrium $\mathbf{E}_1$ is analogous to an infection free equilibrium, in that it is free of the active variable $C$.  Thus, we can obtain an expression for the threshold for CTLs using the next generation matrix method. The term that introduces new CTLs into the system is $\mathcal{F}_1 = \psi I C$ and the term that eliminates them is $\mathcal{V}_1=\omega C$. Let $\mathcal{F} = \left(\frac{\partial \mathcal{F}_1}{\partial C}\right)(\mathbf{E}_1) = \psi I_1$ and $\mathcal{V} = \left(\frac{\partial \mathcal{V}_1}{\partial C}\right)(\mathbf{E}_1)=\omega$. Then, by the next generation matrix method,
\begin{equation}
\hat{R}_1 = \mathcal{F}\mathcal{V}^{-1} = \frac{\lambda \psi}{\delta \omega} \cdot \frac{\beta V_1}{\rho + \beta V_1}.
\label{R1hat}
\end{equation}
Alternatively, following \citet{Schwartzetal2013}, we write the $\mathbf{E}_3$ CTL response (given by eq. \ref{C3}) in the form $C_3=(\delta/k)(R_1 - 1)$ where
\begin{equation}
R_1 = \frac{\lambda \psi}{\delta \omega} \cdot \frac{\beta V_3}{\rho + \beta V_3}
\label{R1}
\end{equation}
Those authors show that $R_1<R_0$ (\citet{Schwartzetal2013}, Theorem 3). $\mathbf{E}_3$  is biological only when $R_1 \geq 1$.

When $\mathbf{E}_1=\mathbf{E}_3$ we have $V_1=V_3$, which implies $R_1 = \hat{R}_1$, and $C_3=C_1=0$, which implies $R_1=1$. On the other hand, if $R_1=\hat{R}_1$, it follows that $V_1=V_3$ and, by virtue of the equations that $V_1$ and $V_3$ solve, {\it i.e.,}

\begin{align}
    \frac{f\alpha}{\mu}V_1^2+\gamma V_1-b I_1&=0,\label{V1Quad}\\
    \frac{f\alpha}{\mu}V_3^2+\gamma V_3-bI_3&=0\label{V3Quad},
\end{align}
that $I_1=I_3$. Since $I_1=\frac{\psi}{\omega}\hat{R}_1=I_3\hat{R}_1$ we obtain $\hat{R}_1=1,$ and therefore that $\mathbf{E}_1=\mathbf{E}_3$. Finally, if $\hat{R}_1=1$, then $I_1=I_3$ and by eqs. \eqref{V1Quad} and \eqref{V3Quad},  $V_1=V_3$. Therefore $R_1=\hat{R}_1$. Thus, the following three statements are equivalent:
\begin{enumerate}
    \item $R_1=\hat{R}_1$,
    \item $\mathbf{E}_1=\mathbf{E}_3$,
    \item $\hat{R}_1=1$.
\end{enumerate}
Therefore, both $R_1$ and $\hat{R}_1$ may be used as a threshold to determine the existence (and stability) of $\mathbf{E}_3$, although only $\hat{R}_1$ should be thought of as a basic reproductive number.

\begin{lemma} If $V_1$ depends on a parameter, that dependence is strictly monotone. \label{lem:transverse}
\end{lemma}
\begin{proof}
By letting $x = \mu/(2\alpha f)$, $y = \rho/(2\beta)$, and $z= \lambda b/(2\delta)$, we can write
\begin{equation*}
    V_1 = -\gamma x-y +\sqrt{(\gamma x+y)^2+4x(z-y\gamma)}.
\end{equation*}
Straightforward calculations show that, for positive parameter values,

\begin{align*}
    \frac{\partial V_1}{\partial y} &=- \frac{x\gamma -y +\sqrt{(x\gamma-y)^2+4xz}}{\sqrt{(x\gamma-y)^2+4xz}}<0,\\  \frac{\partial V_1}{\partial z} &=\frac{2x}{\sqrt{(x\gamma-y)^2+4xz}}>0,\\
    \frac{\partial V_1}{\partial \gamma} &= \frac{x(x\gamma-y-\sqrt{(x\gamma-y)^2+4xz})}{\sqrt{(x\gamma-y)^2+4xz}}<0,\\
    \frac{\partial V_1}{\partial x}&=\frac{2z\gamma V_1}{\gamma(x\gamma-y+\sqrt{(x\gamma -y)^2+4xz})\sqrt{(x\gamma-y)^2+4xz}}.
\end{align*}
The sign of $\frac{\partial V_1}{\partial x}$ is determined by the sign of $V_1$, which is determined by the sign of $R_0 = z/(y\gamma)$, which is independent of $x$. Applying the chain rule gives us the desired result.
\end{proof}
Since the dependence of $\hat{R}_1$ on $V_1$ is strictly monotone, a corollary to Lemma \ref{lem:transverse} is that if $\hat{R}_1$ depends on a parameter, then that dependence is also strictly monotone.

\begin{theorem}
Consider \eqref{eqn:EIAVmodel} with positive parameter values.  The curves of equilibria $\mathbf{E}_0$ and $\mathbf{E}_1$ intersect in a forward transcritical bifurcation when $R_0=1$. The infection free equilibrium $\mathbf{E}_0$ is globally asymptotically stable for $R_0<1$ and unstable for $R_0>1$.  The antibody-only equilibrium $\mathbf{E}_1$ is non-biological for $R_0<1$, is locally asymptotically stable for $\hat{R}_1<1<R_0$, and is unstable if $\hat{R}_1>1$.
\label{E0E1}
\end{theorem}

\begin{proof}
The stability results properties of $\mathbf{E}_0$ follow directly from Theorem \ref{Schwartz}. If $R_0<1$, then it follows from \eqref{eq:V1} that $V_1<0$, and thus $\mathbf{E}_1$ is non-biological. By Lemma \ref{lem:transverse}, the intersection between the curves of equilibria $\mathbf{E}_0$ and $\mathbf{E}_1$ is transverse. When $R_0>1$, the antibody-only equilibrium $\mathbf{E}_1$ is biological.

By writing the equation for CTLs first, the Jacobian at $\mathbf{E}_1$ may be written as
 \begin{equation}
 D\mathbf{g}(\mathbf{E}_1)=\left(
\begin{array}{ccccc}
\omega \hat{R}_1-\omega & 0                 & 0       & 0               & 0 \\
0               & -\rho - \beta V_1 & 0       & -\beta M_1      & 0 \\
-\frac{k\omega}{\psi} \hat{R}_1          & \beta V_1         & -\delta & \beta M_1       & 0 \\
0               & 0                 & b       & -\gamma - f A_1 & -f V_1 \\
0               & 0                 & 0       & \alpha          & -\mu
\end{array} \right),
\label{Dg}
\end{equation}
which has a lower block-triangular form. The eigenvalues of $D\mathbf{g}(\mathbf{E}_1)$ are $\omega(\hat{R}_1-1)$ and the eigenvalues of the lower $4\times4$ matrix. Thus if $\hat{R}_1>1$, then $D\mathbf{g}(\mathbf{E}_1)$ has a positive eigenvalue and $\mathbf{E}_1$ is unstable.
\\When $R_0>1$ the characteristic equation of the lower $4\times4$ matrix satisfy the Routh-Hurwitz criteria (see for example \citet{Meinsma1995}), and therefore have negative real part (for details see the supplementary material). When $R_0 = 1$ there is a transcritical bifurcation that is ``forward'' in the sense that $\mathbf{E}_1$ is locally asymptotically stable for $R_0$ in a neighbourhood of the form $(1,1+\epsilon)$, for $\epsilon$ sufficiently small.
\end{proof}

\begin{theorem}
Consider the EIAV model with positive parameter values and $\lambda \psi > \delta \omega$. The curves of the equilibria $\mathbf{E}_1$ and $\mathbf{E}_3$ intersect in a forward transcritical bifurcation when $\hat{R}_1=1$. The coexistence equilibrium $\mathbf{E}_3$ is non-biological if $\hat{R}_1<1$ and is locally asymptotically stable for $\hat{R}_1$ in a neighborhood of form $(1,1+\epsilon)$ for sufficiently small $\epsilon>0$.
\label{E1E3}
\end{theorem}

\begin{proof}
It is easy to see that the Jacobian at $\mathbf{E}_1$---given by eq. \eqref{Dg}---has a simple zero eigenvalue when $\hat{R}_1=1$. The right and left null vectors of $\left.D\mathbf{g}(\mathbf{E}_1)\right|_{\hat{R}_1=1}$ are $w = (1,w_2,w_3,w_4,w_5)^T$ and $v= (1,0,0,0,0)$ respectively, where
\begin{equation}
w_3 =- \frac{k\lambda(2\alpha f \delta \rho \omega +\beta \gamma \mu(\lambda\psi-\delta\omega)}{\delta^2(\beta\gamma\mu(\lambda \psi-\delta\omega)+f\alpha\rho(\lambda\psi+\delta\omega))},
\end{equation}
which is negative since $\lambda\psi>\delta\omega$. Because of the exact structure of $\left.D\mathbf{g}(\mathbf{E}_1)\right|_{\hat{R}_1=1}$, the expressions for $w_2$, $w_4$, and $w_5$ are not important, but for completion they are shown in the supplementary material. From here, we follow \citet{vandenDriesscheWatmough2002}. Let $p$ be one of the parameters that defines $\hat{R}_1$. By Lemma \ref{lem:transverse}, $\frac{\partial\hat{R}_1}{\partial p} \neq0$, and so the transversality condition,
\begin{equation}
    v\cdot \left.D_{xp}\mathbf{g}(\mathbf{E}_1)\right|_{\hat{R}_1=1}\cdot w = \omega \frac{\partial \hat{R}_1}{\partial p} \neq 0,
\end{equation}
is satisfied. Checking the nondegeneracy condition, we have
\begin{equation}
    v\cdot \left.D_{xx}\mathbf{g}(\mathbf{E}_1)\right|_{\hat{R}_1=1}\cdot w^2 = \psi w_3 <0,
\end{equation}
and thus the bifurcation at $\hat{R}_1 =1$ is a transcritical bifurcation. The sign of $\partial \hat{R}_1/\partial p$ corresponds to the stability of $\mathbf{E}_1$ for $p<p^*$, where $p^*$ is the value of $p$ such that $\left.\hat{R}_1\right|_{p=p^*} =1$. If $\partial \hat{R}_1/\partial p> 0$, then $\mathbf{E}_1$ is stable for $p<p^*$ ($\hat{R}_1<1$), and if $\partial \hat{R}_1/\partial p <0$, then $\mathbf{E}_1$ is unstable for $p<p^*$ ($\hat{R}_1>1$). Thus the transcritical bifurcation is `forward' in the sense that $\mathbf{E}_3$ is biological and locally asymptotically stable for $\hat{R}_1$ in a neighbourhood of the form $(1,1+\epsilon)$ with $\epsilon>0$ sufficiently small.
\end{proof}

{\bf Remark:} Theorem \ref{E1E3} only guarantees stability of $\mathbf{E}_3$ for $\hat{R}_1$ in a neighborhood of the form $(1,1+\epsilon)$ for $\epsilon>0$ sufficiently small. The characteristic equation of the Jacobian at $\mathbf{E}_3$ is a fifth-order polynomial that we were unable to apply the Routh-Hurwitz criterion to. However, we can see from the determininant of the Jacobian,
\begin{equation}
    \det(D\mathbf{g}(\mathbf{E}_3)) = -(\rho+\beta V_3)(R_1-1)(2V_3\alpha f +\gamma \mu),
\end{equation}
that there is only a zero eigenvalue if $R_1=1$, and thus that $\mathbf{E}_3$ does not undergo any further transcritical or saddle-node bifurcations. However, it is not clear that $\mathbf{E}_3$ does not lose stability in a Hopf-bifurcation or in some other, more exotic way. We did not observe any Hopf-bifurcations or more complex dynamics in our numerical exploration.

\section{Numerical Results and Application to EIAV Infection}
\label{NumRes}

In this section, we use bifurcation analysis to explore the EIAV system numerically and determine which parameters play key roles in the system.
Values of immune system parameters $k,\, f,\, \psi,\, \omega,\,$ and $\mu$ were obtained from \citet{Schwartzetal2013}, while the other parameters were taken from a simplified model fitted to data from horses experimentally infected with EIAV\citep{Schwartzetal2018} (Table 1.1).  Equilibrium $\mathbf{E}_0$ corresponds to the infection-free equilibrium (IFE) in which infection does not persist. Boundary equilibrium $\mathbf{E}_1$ corresponds to EIAV infections with only an antibody response, which have not been observed among infected horses.
Interior endemic equilibrium $\mathbf{E}_3$ corresponds to an infection that persists but is controlled at a manageable level by both CTLs and antibodies, which is what is observed in horses \citep{Lerouxetal2004}. For this set of parameters, we have $\lambda\psi>\delta\omega$ and Theorem \ref{E1E3} applies.

{\begin{center}Table 1.1: Parameter Values Used In Numerical Results
\begin{table}[h]

\label{TabA}       
\begin{tabular}{llll}
\hline\noalign{\smallskip}
Symbol & Definition & Value & Units  \\
\noalign{\smallskip}\noalign{\smallskip}\hline\noalign{\smallskip}
$\alpha$ & antibody production rate & 15 & molecules/(vRNA$\cdot$day) \\
$\beta$ & infectivity rate & 0.000325 & $\mu$l/(vRNA$\cdot$day) \\
$\gamma$ & virus clearance rate & 6.73 & 1/day \\
$\delta$ & infected cell death rate & 0.0476 & 1/day \\
$\lambda$ & uninfected cell arrival rate & 2.019 & cells/($\mu$l$\cdot$day) \\
$\mu$ & antibody clearance rate & 20 & 1/day \\
$\rho$ & uninfected cell death rate & 0.0476 & 1/day \\
$\psi$ & CTL production rate & 0.75 & $\mu$l/(cell$\cdot$day) \\
$\omega$ & CTL death rate & 5 & 1/day \\
$b$ & virus production rate & 505 & vRNA/(cell$\cdot$day) \\
$f$ & antibody neutralization rate & 3 & $\mu$l/(molecule$\cdot$day) \\
$k$ & rate of killing by CTLs & 0.01 & $\mu$l/(cell$\cdot$day) \\
\noalign{\smallskip}\hline
\end{tabular}%
\caption{Parameter Values Used In Numerical Results}
\end{table}
\end{center}}

Figure \ref{fig:betaV} uses the viral infectivity $\beta$ as the control parameter, and the virus particle concentration $V$ is shown on the vertical axis.  The curve of IFE $\mathbf{E}_0$ intersects the curve of boundary equilibrium $\mathbf{E}_1$ when the basic reproduction number $R_0=1$.  This occurs when $\beta = \beta_0 \approx 1.50 \times 10^{-5} \mu$l/(vRNA$\cdot$day).  There is a forward transcritical bifurcation at the intersection; the infection free equilibrium $\mathbf{E}_0$ is stable for $\beta<\beta_0$ and unstable for $\beta>\beta_0$. The antibody-only equilibrium $\mathbf{E}_1$ is non-biological for $\beta<\beta_0$ and is stable for $\beta>\beta_0$. As the infectivity $\beta$ increases beyond $\beta_0$, the equilibrium virus concentration $V$ increases.  Another forward transcritical bifurcation takes place when $\beta = \beta_1\approx2.38\times10^{-4}\mu$l/(vRNA$\cdot$day). The antibody-only equilibrium $\mathbf{E}_1$ is stable for $\beta<\beta_1$ and unstable for $\beta>\beta_1$.  The coexistence equilibrium $\mathbf{E}_3$, in which the antibody and CTL responses coexist, is non-biological for $\beta<\beta_1$ and is stable for $\beta>\beta_1$. As the infectivity $\beta$ increases beyond $\beta_1$, the equilibrium virus concentration $V$ remains steady.

\begin{figure}[ht]
    \centering \includegraphics[width=\textwidth]{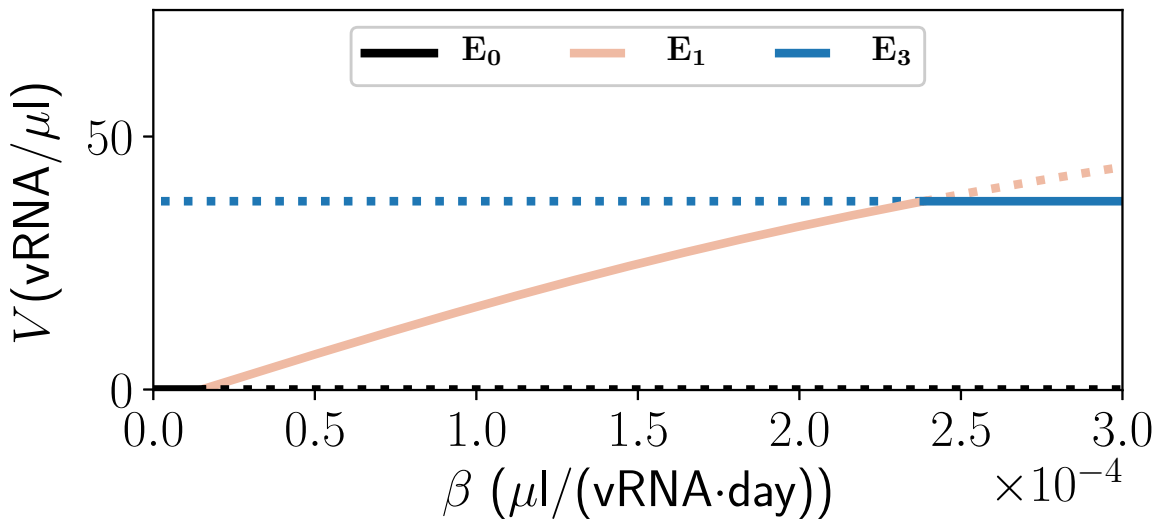}
    \caption{ Equilibrium values of the virus concentration $V$ as a function of the viral infectivity $\beta$.  The solid lines indicate when an equilibrium is stable, while the dotted lines indicate when the equilibrium is unstable. When $\mathbf{E}_3$ is unstable, it is also non-biological, and when $\mathbf{E}_3$ is stable, it is biological. Parameter values (except for $\beta$) are given in Table 1.1.}
    \label{fig:betaV}
\end{figure}

Similarly, Figure \ref{fig:bV} shows the equilibrium virus concentration using the virus production rate $b$ as the control parameter. There is a forward transcritical bifurcation at the intersection of the IFE $\mathbf{E}_0$ and boundary equilibrium $\mathbf{E}_1$ curves, which occurs when $b = b_0\approx 20$ vRNA/(cell$\cdot$day). Another forward transcritical bifurcation occurs at the intersection of the boundary equilibrium $\mathbf{E}_1$ and interior equilibrium $\mathbf{E}_3$ curves, which occurs at $b=b_1\approx 213$ vRNA/(cell$\cdot$day).

Figures \ref{fig:betaV} and \ref{fig:bV} show how different values of parameters $\beta$ and $b$ drive the system to different equilibria.  Lower viral infectivity ($\beta$) and lower production of virus ($b$) correspond to lower levels of virus $V$ and stability of antibody-only equilibrium $\mathbf{E}_1$.  Alternatively, greater infectivity ($\beta$) and greater virus production ($b$) give higher virus levels $V$ and stability of the coexistence equilibrium $\mathbf{E}_3$.

\begin{figure}[ht]
    \centering
\includegraphics[width=\textwidth]{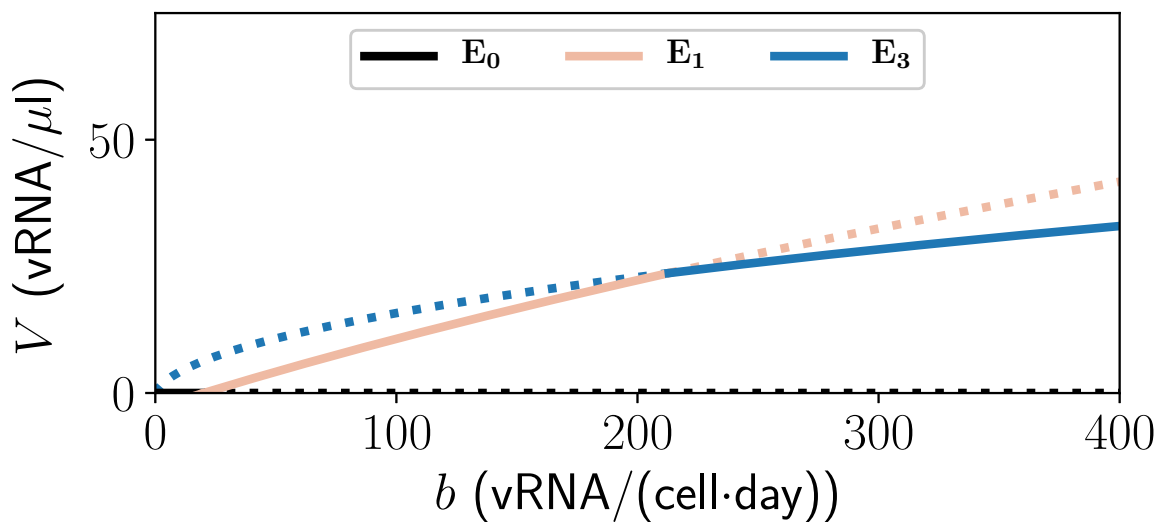}
    \caption{Equilibrium values of the virus concentration $V$ as a function of the virus production rate $b$.  The solid lines indicate when an equilibrium is stable, while the dotted lines indicate when the equilibrium is unstable. When $\mathbf{E}_3$ is unstable, it is also non-biological. Parameter values (except for $b$) are given in Table 1.1.}
    \label{fig:bV}
\end{figure}

Figure \ref{fig:betaB} shows a two-parameter bifurcation diagram using viral infectivity $\beta$ and the virus production rate $b$ as bifurcation parameters. The first transcritical bifurcation occurs when $R_0 = 1$. Since $R_0$ depends on both $\beta$ and $b$, it is possible to solve for $b$ as a function of $\beta$, and the bifurcation appears as a decreasing curve in the $(\beta,b)$ plane, separating the region where $\mathbf{E}_0$ is stable from the region where $\mathbf{E}_1$ is stable.  The second transcritical bifurcation occurs when $R_1 = 1$ and separates the region where $\mathbf{E}_1$ is stable from the region where $\mathbf{E}_3$ is stable.
Note that if the value of the infectivity ($\beta$) or virus production rate ($b$) is low enough, then IFE $\mathbf{E}_0$ can be reached with a broad range of values in the other parameter.

\begin{figure}[h]
\centering
    \includegraphics[width=\textwidth]{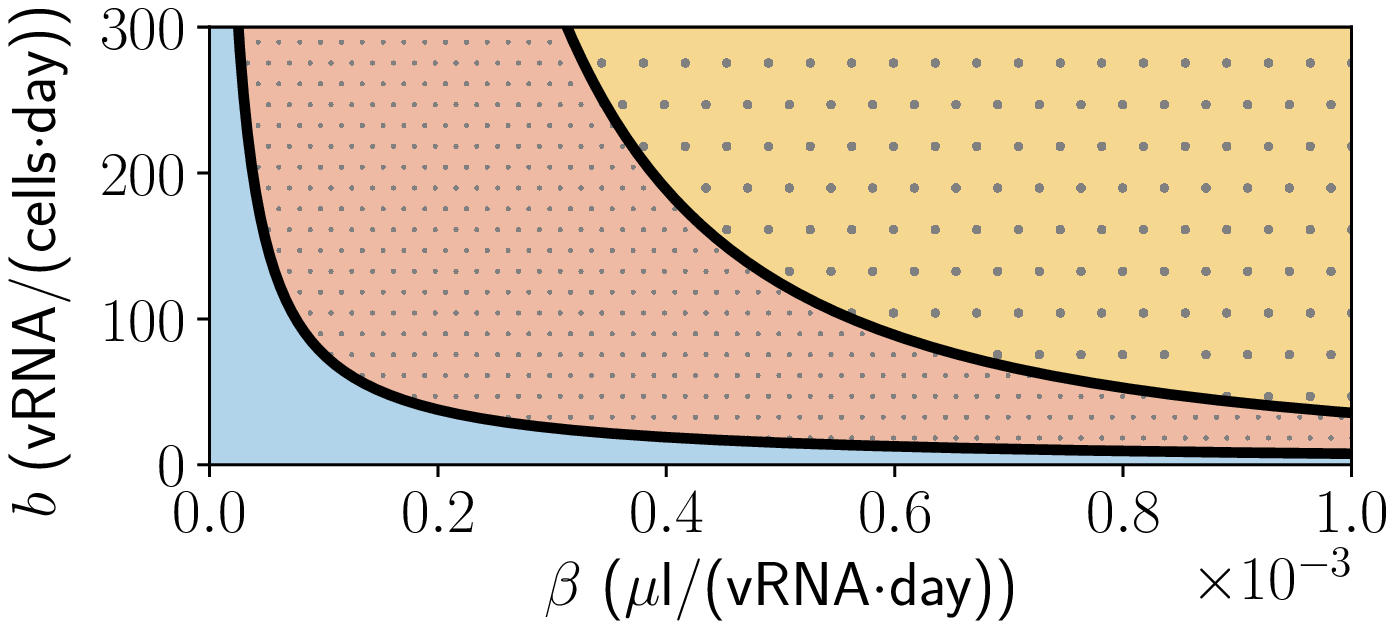}
    \caption{Two-parameter bifurcation diagram showing how the critical values of the virus production rate ($b$) depend on infectivity ($\beta$). The black curve on the left represents the first transcritical bifurcation, when $R_0=1$. In the blue area to the left of and below this line, the IFE $\mathbf{E}_0$ is stable. The black curve on the right represents the second transcritical bifurcation, when $R_1=1$. In the tightly dotted red region between these two curves, the boundary equilibrium $\mathbf{E}_1$ is stable, and in the loosely dotted yellow region to the right of this curve, the endemic equilibrium $\mathbf{E}_3$ is stable. Parameter values (other than $\beta$ and $b$) are taken from Table 1.1.}
    \label{fig:betaB}
\end{figure}

The implication of the results shown in Figures \ref{fig:betaV}, \ref{fig:bV} and \ref{fig:betaB} is that modification of the infectivity $\beta$ or virus production rate $b$, such as by using antiretroviral therapies (ART) that block infection of cells or inhibit production of virus, respectively, can theoretically shift the system to a more preferred state, such as the IFE  $\mathbf{E}_0$.
ART, however, is not used in practice for treating EIAV-infected horses, whose immune systems manage the persistent infection without symptoms throughout most of their lives.  Thus, we next investigate how to shift the system's equilibria by modifying the immune system parameters. In general, immune responses can be boosted by vaccination. Consequently, we examine the CTL production rate $\psi$ and the antibody production rate $\alpha$.

Figure \ref{fig:psi}~
shows the $\mathbf{E}_1$-$\mathbf{E}_3$ bifurcation with CTL production rate $\psi$ as the control parameter.

The antibody-only equilibrium $\mathbf{E}_1$ is stable for $\psi<\psi_1$ and is unstable for $\psi>\psi_1$, with $\psi_1 \approx 0.48 \mu$l/(cell$\cdot$day).   The coexistence equilibrium $\mathbf{E}_3$ is non-biological for $\psi<\psi_1$ and is stable for $\psi>\psi_1$. As the value of $\psi$ increases above $\psi_1$, the equilibrium CTL concentration $C$ increases above zero (Fig. \ref{fig:psiC}) and the equilibrium infected cell concentration $I$ decreases (Fig. \ref{fig:psiI}).  This is consistent with the known function of CTLs, whose role is killing infected cells. In other words, the greater the production of CTLs, the higher the CTL level and the lower the number of infected cells. Equilibrium $\mathbf{E}_3$ is characterized by the presence of both CTLs and antibodies, which is the condition that gives rise to control of virus infection in EIAV-infected horses.
\begin{figure}[ht]
   \begin{subfigure}{0.5\textwidth}   \includegraphics[width=\textwidth]{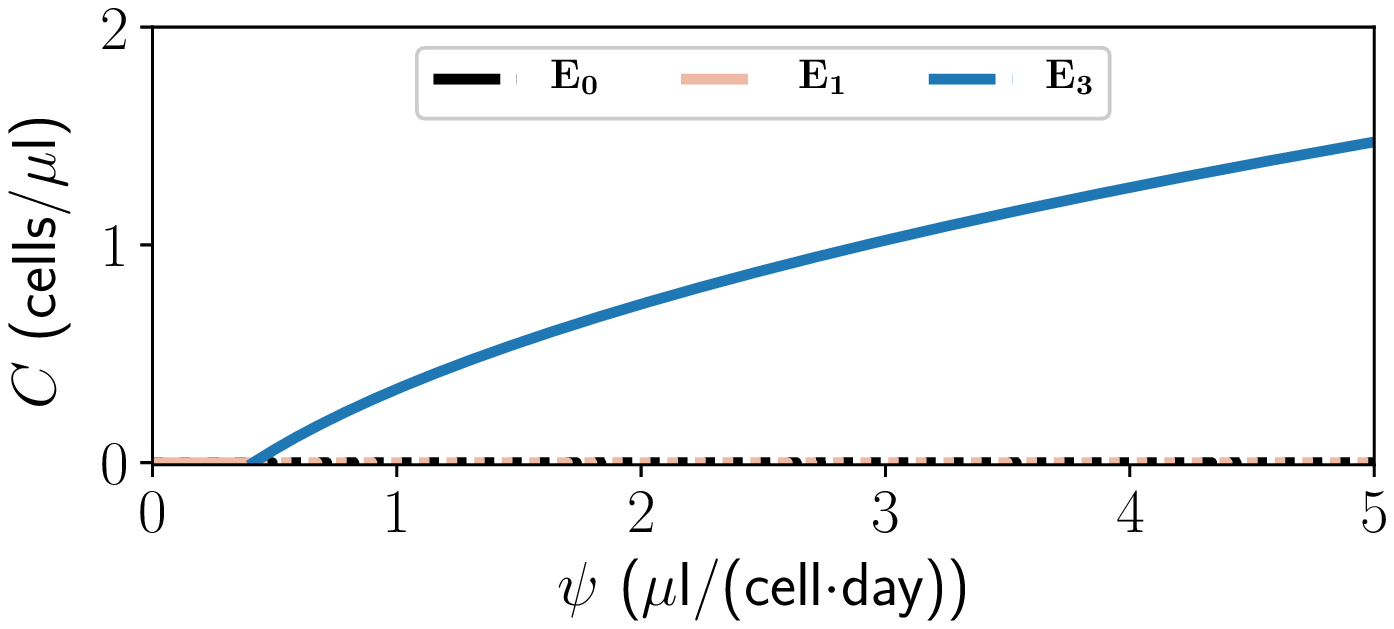}\caption{}\label{fig:psiC}
   \end{subfigure}
   \begin{subfigure}{0.5\textwidth}
   { \includegraphics[width=\textwidth]{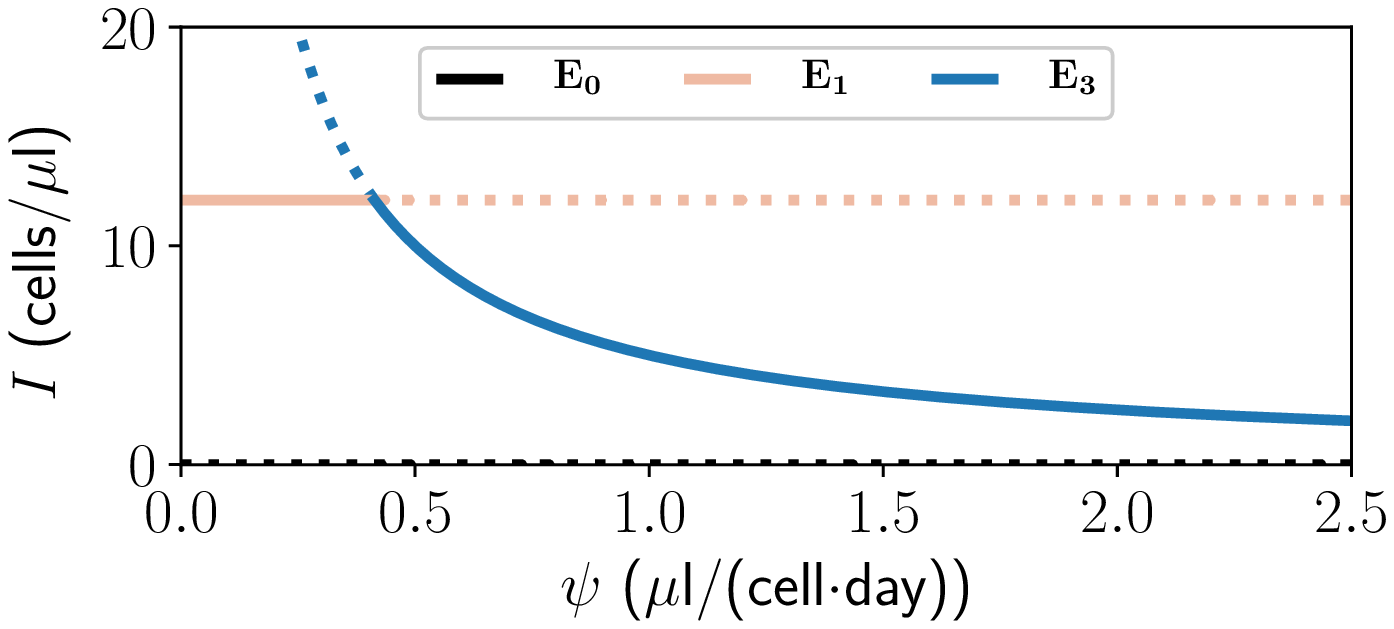}\caption{}\label{fig:psiI}}
   \end{subfigure}
    \caption{(a) Equilibrium values of the CTL concentration $C$ as a function of the CTL production coefficient $\psi$. (b) Equilibrium values of the infected cell concentration $I$ as a function of the CTL production coefficient $\psi$.  The solid lines indicate when an equilibrium is stable, while the dotted lines indicate when the equilibrium is unstable. When $\mathbf{E_3}$ is unstable, it is also non-biological. Parameter values (except for $\psi$) are given in Table 1.1. \label{fig:psi}}
\end{figure}

Figure \ref{fig:alpha}
shows the $\mathbf{E}_1$-$\mathbf{E}_3$ bifurcation with antibody production rate $\alpha$ as the control parameter and virus particle concentration $V$ (Fig. \ref{fig:alphaV}) and antibody concentration $A$ (Fig. \ref{fig:alphaA}) on the vertical axes.  When the antibody production rate ($\alpha$) takes on lower values (less than $\approx$38 molecules/(vRNA$\cdot$day)), the equilibrium viral load $V$ is higher (Fig. \ref{fig:alphaV}), the equilibrium antibody level $A$ is lower (Fig. \ref{fig:alphaA}), and the system is driven to stability of $\mathbf{E}_3$ (characterized by the existence not only of antibodies but also CTLs). When antibody production $\alpha$ takes on higher values (greater than $\approx$38 molecules/(vRNA$\cdot$day)), the equilibrium viral load $V$ is lower (Fig. \ref{fig:alphaV}), the equilibrium antibody level $A$ plateaus (Fig. \ref{fig:alphaA}), and the antibody-only equilibrium $\mathbf{E}_1$ is stable.  This result suggests that an antibody response that moderately reduces, but does not strongly reduce, the virus, is consistent with stability of coexistence equilibrium $\mathbf{E}_3$.

\begin{figure}[ht]
\centering
\begin{subfigure}{0.48\textwidth}
    \includegraphics[width=\textwidth]{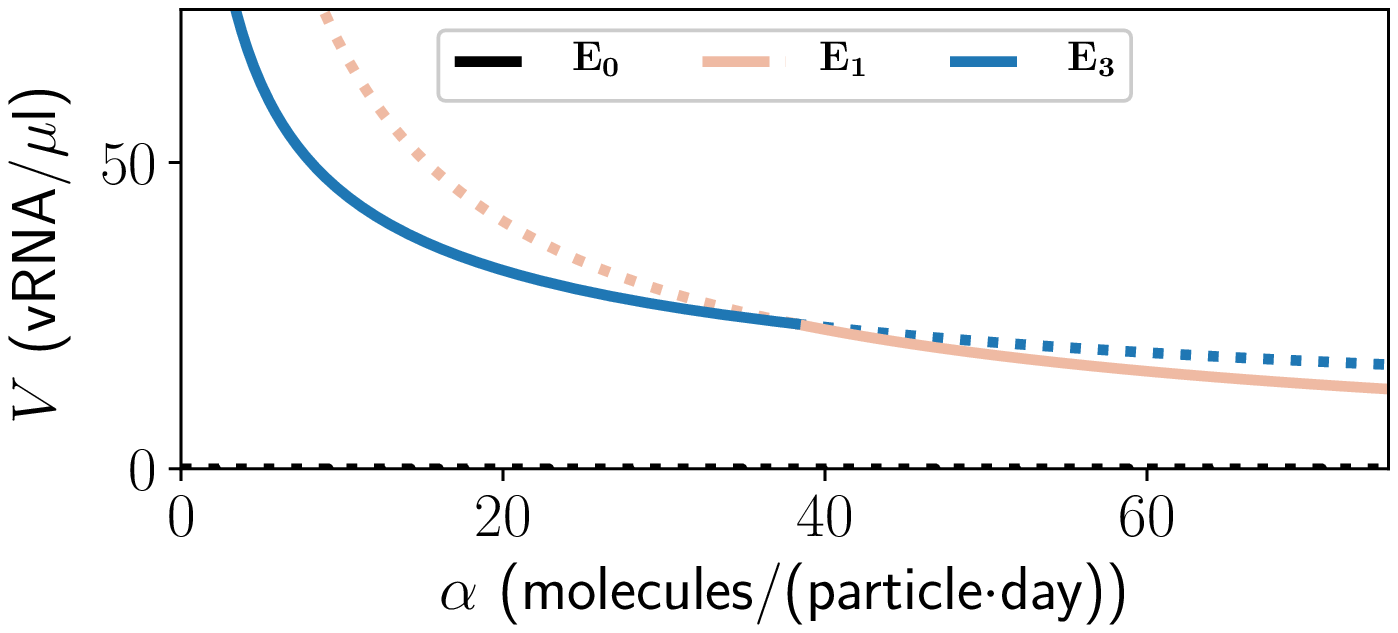}\caption{}
    \label{fig:alphaV}
    \end{subfigure}
    \begin{subfigure}{0.48\textwidth}
    \includegraphics[width=\textwidth]{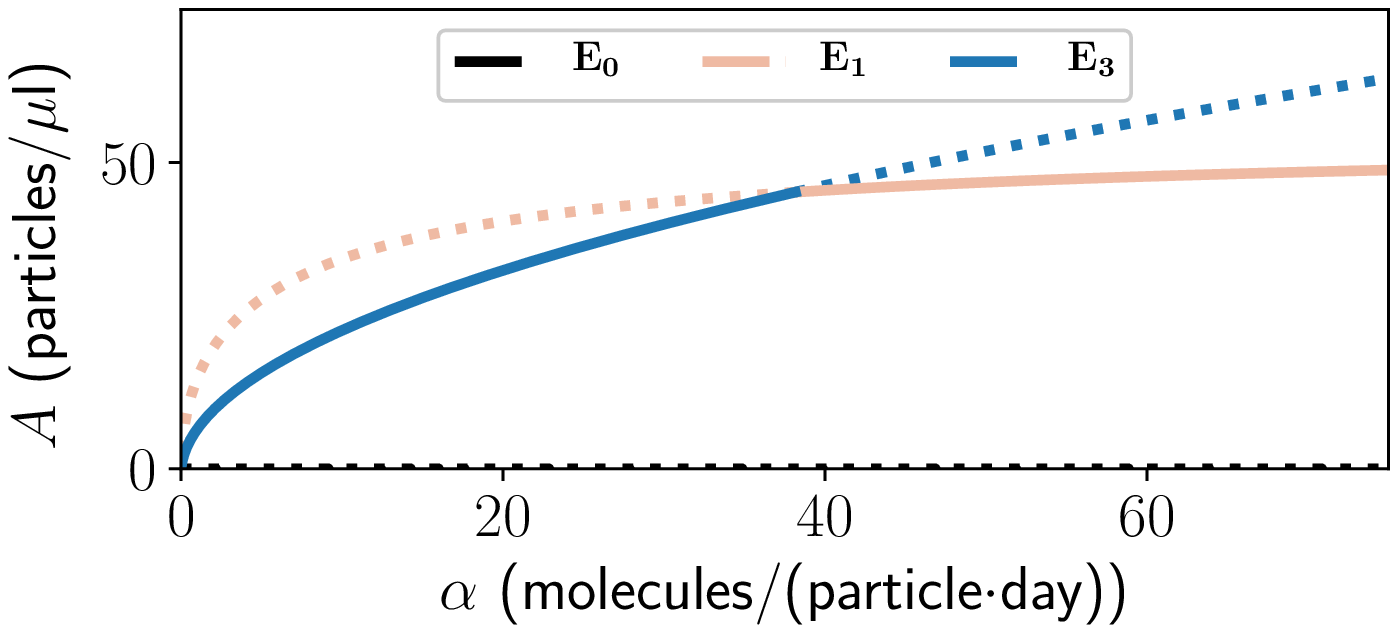}\caption{}
    \label{fig:alphaA}
    \end{subfigure}
    \caption{(a) Equilibrium values of the virus concentration $V$ as a function of the antibody production rate $\alpha$. (b) Equilibrium values of the antibody concentration $A$ as a function of the antibody production rate $\alpha$.  The solid lines indicate when an equilibrium is stable, while the dotted lines indicate when the equilibrium is unstable. When $\mathbf{E}_3$ is unstable, it is also non-biological. Parameter values (except for $\alpha$) are given in Table 1.1.} \label{fig:alpha}
\end{figure}

Figure \ref{fig:alphapsi}~
shows a two-parameter bifurcation diagram using the antibody production rate $\alpha$ and the CTL production rate $\psi$ as bifurcation parameters.The second transcritical bifurcation, occurring when $R_1 = 1$, is an increasing curve in the $\alpha, \psi$ plane, separating the region where $\mathbf{E}_1$ is stable from the region where $\mathbf{E}_3$ is stable.  Since $R_1$ depends on both $\psi$ and $\alpha$, it is possible to solve for $\psi$ as a function of $\alpha$, appearing almost as an inverse relationship, where higher values of $\alpha$ require even higher values of $\psi$ for the stability of $\mathbf{E}_3$.
\begin{figure}[h]
    \centering
    \includegraphics[width=0.98\textwidth]{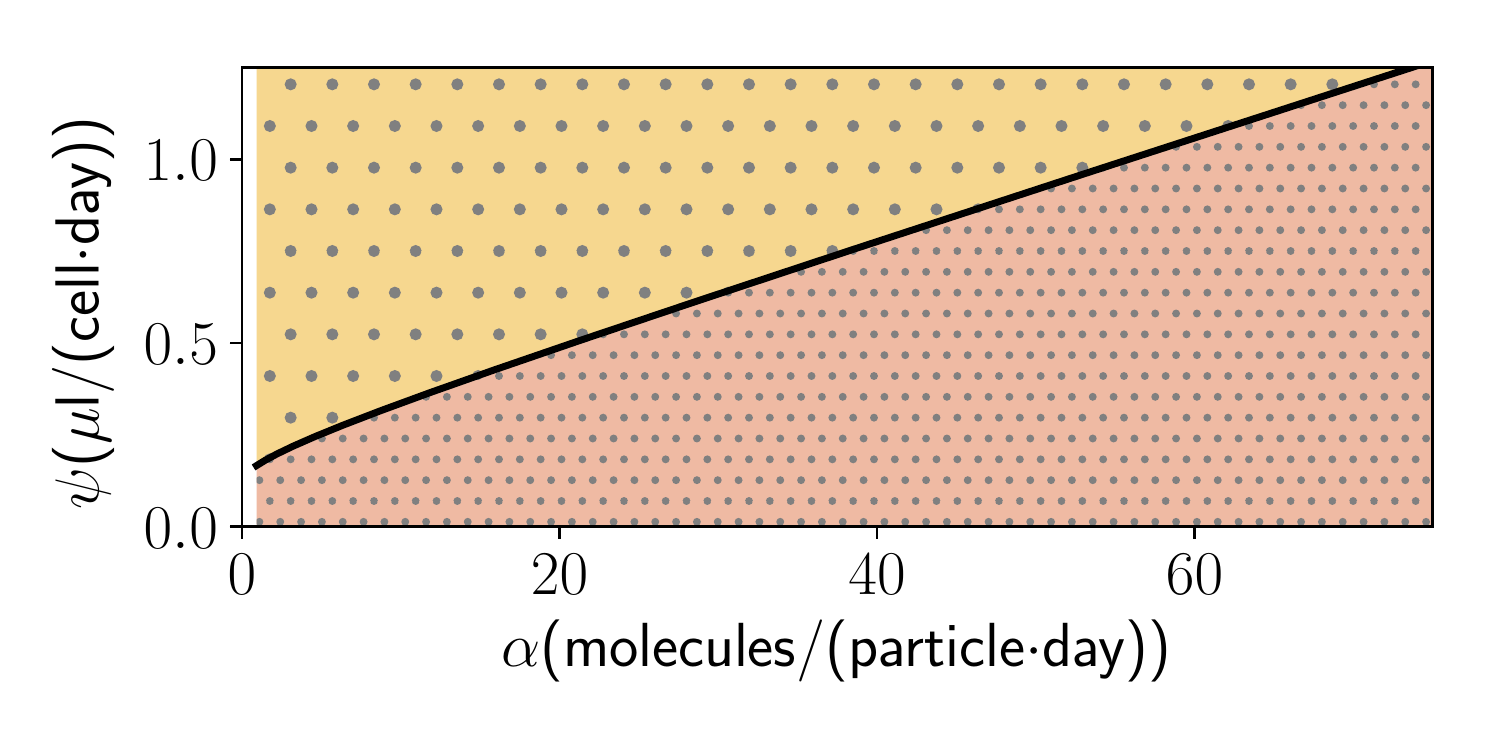}
    \caption{Two parameter bifurcation diagram showing how the critical value of the CTL production rate $\psi$ depends on the antibody production rate $\alpha$. The curve represents the second transcritical bifurcation, when $R_1=1$. In the tightly dotted red area below the curve, the boundary equilibrium $\mathbf{E}_1$ is stable, and in the loosely dotted yellow area above this curve, the endemic equilibrium $\mathbf{E}_3$ is stable. Parameter values (other than $\alpha$ and $\psi$) are taken from Table 1.1.}
    \label{fig:alphapsi}
\end{figure}

This result suggests that for stability of coexistence equilibrium $\mathbf{E}_3$,
the required value of the CTL production rate $\psi$ increases with increasing antibody production rate $\alpha$.  In other words, any increase in antibody production should be coupled with an increase in CTL production; otherwise the system is driven toward $\mathbf{E}_1$, with antibody responses and no CTLs.  In summary, these numerical results may seem somewhat counter-intuitive:
To obtain stability of $\mathbf{E}_3$, characterized by the coexistence of both antibody and CTL responses, the desired immune response has lower $\alpha$ (i.e., less production of antibodies), and greater $\psi$ ({\it i.e.,} greater production of CTLs).  A very strong antibody production ({\it i.e.,} high $\alpha$) associated with $\psi<\psi_1$, however, correlates with an absence of CTLs.

\section{Discussion}

In this paper we analyze the equilibrium states of a virus infection model with immune system responses in the form of antibodies and cytotoxic T lymphocytes (CTLs). Using a standard Lyapunov function argument, we show that the infection free equilibrium $\mathbf{E}_0$ is globally asymptotically stable when the basic reproductive number $R_0$ is less than one. When $R_0=1$ there is a forward transcritical bifurcation where the infection free equilibrium loses stability to the boundary equilibrium $\mathbf{E}_1$, which describes an infection that is controlled by antibodies but not CTLs. Using the next generation matrix method by \citet{Diekmannetal1990}, \citet{vandenDriesscheWatmough2002}, and \citet{vandenDriessche2017}, we derive a reproductive number for CTLs, $\hat{R}_1$, and show that $\mathbf{E}_1$ is locally asymptotically stable when $\hat{R}_1<1<R_0$. When $\hat{R}_1=1$ there is a second forward transcritical bifurcation where the boundary equilibrium loses stability to the endemic equilibrium $\mathbf{E}_3$, which describes an infection that is controlled by both antibodies and CTLs. We are unable to show that $\mathbf{E}_3$ remains locally stable as $\hat{R}_1$ increases, but our numerical analysis suggests that this is the case.

Our results are similar to those of \citet{Gomez-AcevedoLi2010}, who examine a three-equation model of infection by Human T cell Leukemia/Lymphoma virus, HTLV.  They obtain a basic reproduction number $R_0$ for infected cells and a second threshold $R_1$, which they also interpret as a basic reproduction number for CTLs.

Their theorem 3.1 corresponds to our results relating the stability of the equilibria to the basic reproduction numbers $R_0$ and $\hat{R}_1$,
consistent with our analysis (of the bifurcations between $\mathbf{E}_0$, $\mathbf{E}_1$, and $\mathbf{E}_3$) showing that $\hat{R}_1$ is precisely the basic reproduction number of CTLs.

The numerical results presented here offer insights into the potential relationships between immune response parameters in EIAV infection.  Furthermore, such insights have implications for vaccine development. For instance, an antibody response that moderately, but not strongly, reduces virus is consistent with stability of the coexistence equilibrium $\mathbf{E}_3$ ({\it i.e.,} control of virus infection with CTLs and antibodies). A vaccine, therefore, that aims to stimulate antibody production modestly would drive the system to this state of control of infection.  This would lead to a lower antibody concentration and a higher virus concentration, which may seem counter-intuitive, since conventional wisdom would presume that a vaccine that stimulates more antibody production would lead to greater reduction of virus and more control. However, the results shown here describe how the antibody response must be tempered in order to allow for the coexistence of a CTL response. Thus, a potential vaccine that stimulates the production of antibodies would need to stimulate the production of CTLs as well.  Overstimulation of antibodies would drive the system to the equilibrium state devoid of CTLs. Consequently, a vaccine intended to control virus infection by stimulating both immune responses would aim to favor the CTL response in order to balance the antibody response.  An ideal vaccine
would accomplish this balance.

Some limitations of the work presented here should be discussed.  In this work, we are motivated by the finding that EIAV infection (unlike HIV) is controlled by the host adaptive immune response, and this control is mediated by both antibody and CTL responses.  Thus our goal here is to use the knowledge that both responses are crucial for control as the basis to explore the asymptotic behavior of a model that considers each response’s dynamics separately. However, other models more explicitly describe the clonal expansion of the antibody response as well as the kinetics of CTL growth \citep{antia2005role}. Future work that addresses these modeling hurdles will help understand the role of complex immune responses.

Another caveat of this work is a reliance on deterministic population dynamics; stochastic interactions are not taken into account in this model.  In addition, this study does not consider spatial heterogeneity or diffusion.

Other studies in the literature do consider stochasticity in within-host dynamics \citep{Gibelli2017, SchwartzYangCumberlanddePillis2013}, as well as spatial heterogeneity and diffusion \citep{Gibelli2017, Bellomo2019}. \citet{Gibelli2017} expands upon the basic model \citep{NowakBangham1996, Perelson2002} by including population heterogeneity (in this case, in the age-distributed time of cell death and variation in the timing of eclipse phase dynamics).  \citet{Bellomo2019} takes into account spatial effects of the three populations of the basic model ({\it i.e.,} uninfected cells, infected cells, and virus), particularly the contributions of diffusion and movement by chemotaxis.  While research on stochastic dynamical systems shows that the deterministic structure of a model is often still strongly apparent with the addition of stochasticity \citep{Abbott2017}, future studies that use hybrid models that include stochastic and deterministic dynamics, and modeling that considers heterogeneity, will advance the field by leading to more precise depictions of the biological scenarios being modeled. Our model and analysis presented here may form the foundation of such future work.

\section{Acknowledgements}
We would like to thank Mark Schumaker for substantial input on an earlier draft of the paper. We would also like to thank Christina Cobbold, Adriana Dawes, Abba Gumel, Fabio Milner, Stacey Smith?, Rebecca Tyson, and Gail Wolkowicz for their suggestions on the mathematical and numerical analyses and conversations about this work.  We also thank two anonymous reviewers for their suggestions, which improved the paper.

This work was partially supported by the Simons Foundation and partially supported by the National Institute of General Medical Sciences of the National Institutes of Health under Award Number P20GM104420. The content is solely the responsibility of the authors and does not necessarily represent the official views of the National Institutes of Health.

\section{Appendices}
\subsection{Supplementary Material}

This section briefly describes the supplementary materials, which are Jupyter notebooks that fill in some of the details of the proofs in this paper. The DOI provides the link to this material online. \\

\noindent {\bf Notebook 1}:  In this notebook we develop the Routh-Hurwitz conditions for the stability of a fourth order polynomial by constructing the table described by \citet{Meinsma1995}. https://doi.org/10.7273/000002580 \\

\noindent {\bf Notebook 2}: In this notebook we calculate the characteristic polynomial associated with the Jacobian matrix $D\mathbf{g}(\mathbf{E}_1)$ given by eq. \ref{Dg}. We cast the coefficients of the characteristic polynomial in forms that are manifestly positive, and verify the Routh-Hurwitz criterion. The third order Routh-Hurwitz criterion comes out as a sum of 172 terms, of which two are negative. We show that squares can be completed, combining the negative terms with other terms in a form that is positive. https://doi.org/10.7273/000002581 \\

\noindent {\bf Notebook 3}: In this notebook we calculate the left and right nullvectors of the Jacobian matrix $D\mathbf{g}(\mathbf{E}_1)$ given by eq. \ref{Dg}.  We then cast the nondegeneracy condition in a form that is manifestly negative. https://doi.org/10.7273/000002582

\bibliographystyle{plainnat}
\bibliography{eiav}

\end{document}